\documentclass[prd,showpacs,preprintnumbers,groupedaddress]{revtex4-1}
\usepackage{amsthm}
\theoremstyle{plain}
\newtheorem*{theorem}{Theorem}
\newtheorem*{definition}{Definition}
\newtheorem{lemma}{Lemma}
\newtheorem{corollary}{Corollary}
\usepackage{cases}
\usepackage{color}
\usepackage{comment}
\usepackage[normalem]{ulem}%for the way of expression of bib
\newcommand{\llaabel}[1]{\label{#1}}

\begin{document}

%\preprint{}

\title{Correspondence between sonic points of ideal photon gas accretion and photon spheres}


\author{Yasutaka Koga}
%\thanks{}
\author{Tomohiro Harada}
%\thanks{}
\affiliation{Department of Physics, Rikkyo University, Toshima, Tokyo 171-8501, Japan}


\date{\today}

\begin{abstract}
In the accretion flow of fluid, its velocity may transit from subsonic to supersonic.
The point at which such transition occurs is called sonic point and often mathematically special.
We consider the steady-state and spherically symmetric accretion
 problem of ideal photon gas in general static spherically symmetric spacetimes neglecting back reaction.
Our main result is that the equation of state (EOS) of ideal photon gas leads to
 correspondence between its sonic point and the photon sphere of the
 spacetime in general situations. Moreover, we also show that
in spite of the dependence of the EOS on the dimension of spacetime,
this correspondence holds for spacetimes of arbitrary dimensions.
\end{abstract}

\pacs{04.20.-q, 04.40.Nr, 98.35.Mp}
%\keywords{}

\maketitle

%\section{}
%\subsection{}
%\subsubsection{}

%%%chap1
\section{Introduction}
\llaabel{sec:introduction}
A photon sphere is a sphere on a spacetime at which circular null geodesics exist.
In astrophysical cases, black holes usually have photon spheres near the horizon.
This structure is the characteristics of strong gravitational fields and helps us to identify black holes in the Universe by optical observations through its gravitational lensing.
The size of the shadow of the hole is determined by the radius of its photon sphere. 
In the case of the Schwarzschild black hole for example, we can see their relation from the calculation  by Synge~\cite{synge}.
\par
The accretion of fluid onto objects is a basic
problem in astrophysics and
the most important issue concerning growth of stars and black holes.
In an observational view point, the accretion is considered to be responsible
for the X-ray emission due to the compression of the fluid.
This is also connected to the observations of strong gravity fields in a general relativistic context.
\par
The first study of the accretion onto stars was established by Bondi~\cite{bondi}.
He investigated stationary spherically symmetric flow of polytropic fluid in Newtonian gravity.
One of the interesting features is the existence of a critical point 
(or sonic point) and transonic flow,
that is, flow which experiences transition between subsonic and supersonic states.
Michel extended the problem to general relativity on the 
Schwarzschild spacetime
under the assumption that the spacetime is not so strongly modified by the fluid and
also estimated several quantities on the critical point~\cite{michel}.
For the (anti-)de Sitter spacetime, Mach, Malec and Karkowski gave not only numerical calculations with a 
polytropic equation of state (EOS), but also the exact solutions of the accretion of fluid with isothermal EOSs~\cite{mach}.
For general static spherically symmetric spacetimes and polytropic EOSs, the existence of the unique solution of the accretion problem has been proved
by Chaverra and Sarbach~\cite{chaverra}.
They analyzed the problem by the 
method of dynamical systems.
In the analysis of the outflow of fluid, there exist the same features,
i.e., a transonic flow and a sonic point as in the accretion problem.
Carter, Gibbons, Lin and Perry discussed the treatment
of Hawking radiation from astrophysical black holes 
as the outflow of perfect fluid~\cite{Carter}.
\par
In the study by Mach et. al~\cite{mach}, it was revealed that only for the case of the accretion of radiation fluid, 
the radius of the sonic point is $3M$.
This radius coincides with the photon sphere of the spacetime.
This correspondence connects between two independent observations,
the observation of lights from sources behind a black hole and the observation of emission from
accreted radiation fluid onto the hole, because the size of the shadow of the hole is determined by the radius of the photon sphere and the accreted fluid can signal the sonic point.
\par
In this paper, we see there exists the correspondence between the
sonic points of photon gas accretion and the photon spheres in large class by generalizing the analysis~\cite{chaverra} to arbitrary dimensions.
In fact, we consider general static spherically symmetric spacetimes in $D$ dimensions
\begin{equation}
\llaabel{eq:metric}
ds^2=-f(r)dt^2+g(r)dr^2+r^2d\Omega^2_{D-2},
\end{equation}
where $D\ge 3$ and the condition $0<f,g<\infty$ are assumed and $d\Omega^2_{D-2}$ is the unit $(D-2)$-sphere metric.
The main result is:
\begin{theorem}\llaabel{theorem:main}
For any physical transonic accretion flow of ideal
 photon gas in stationary and spherically symmetric state on the fixed
 background spacetime (\ref{eq:metric}), the
radius of its sonic point coincides with that of (one of) the unstable photon sphere(s) of the geometry.
\end{theorem}
The sonic point is a point at which transition between supersonic and subsonic states occurs.
The term unstable photon sphere means the instability of the corresponding circular orbits of null geodesics.
The rigorous definitions of ``physical flow'' and the other terms will be given in the following sections.
\par
In Sec.~\ref{sec:photonsphere}, we derive the conditions for the radius of photon sphere of the spacetime and its stability.
In Sec.~\ref{sec:accretionproblem}, we formulate the general accretion problem of stationary and spherically symmetric accretion on the $D$ dimensional spacetime.
Also are the critical point and the sonic point defined.
In Sec.~\ref{sec:photongas}, we introduce the EOS of ideal photon gas in $d$ dimensional space.
Then the critical point of the ideal photon gas accretion in $D$
dimensions of spacetime is obtained.
In Sec.~\ref{sec:proof}, the main theorem is proved and conclusion is
given by Sec.~\ref{sec:conclusion}.

%%%chap2
\section{The photon sphere}
\llaabel{sec:photonsphere}
A photon sphere is defined as a sphere on which circular null geodesics exist ~\cite{chandrasekhar}.
A photon sphere is said to be stable and unstable, if it has stable and 
unstable circular orbits, respectively. 
We present the following lemma for the photon sphere of the spacetime (\ref{eq:metric}).
\begin{lemma}\llaabel{lemma:photonsphere}
Let the metric be Eq.~(\ref{eq:metric}). 
The photon sphere of the spacetime is specified by the equation
\begin{equation}
\llaabel{eq2}
(fr^{-2})'=0.
\end{equation}
The stability condition of the photon sphere is given by
\begin{equation}
\llaabel{eq:spherestability}
stable\ (unstable) \Leftrightarrow (fr^{-2})''>0 \ (<0)
\end{equation}
at the radius of the photon sphere.
\end{lemma}
\begin{proof}
Consider a null geodesic $x^\mu=x^\mu(\lambda)$ confined in $\theta=\frac{\pi}{2}$ surface where $\lambda$ is the affine parameter.
The null condition leads to the equation
\begin{equation}
\llaabel{eq:hamiltonian}
\mathcal{H}=\frac{1}{2}g_{\mu\nu}\dot{x}^\mu\dot{x}^\nu=0
\end{equation}
for its Hamiltonian $\mathcal{H}$, where $\dot{ }=d/d\lambda$.
From two Killing vectors relevant to the motion,
\begin{equation}
\xi_{(t)}=\partial_t,\   \xi_{(\phi)}=\partial_\phi,
\end{equation}
we have two conserved quantities,
\begin{eqnarray}
E:&=&-g_{\mu\nu}\xi^\mu_{(t)}\dot{x}^\nu,\\
L:&=&g_{\mu\nu}\xi^\mu_{(\phi)}\dot{x}^\nu,
\end{eqnarray}
and the Hamiltonian reduces to
\begin{eqnarray}
\llaabel{eq:kplusv}
\mathcal{H}&=&\frac{1}{2}g\dot{r}^{2}-\frac{1}{2f}[E^{2}-L^{2}fr^{-2}]\nonumber\\
 &=&g\left[\frac{1}{2}\dot{r}^{2}+V(r)\right],
%\frac{1}{2}\dot{r}^2-\frac{1}{2fg}\left[E^2-L^2fr^{-2}\right]\nonumber \\
%&=:&\frac{1}{2}\dot{r}^2+V(r).
\end{eqnarray}
where
\begin{equation}
V(r):=-\frac{1}{2fg}[E^{2}-L^{2}fr^{-2}].
\end{equation}
Defining $F(r):=E^2-L^2fr^{-2}$, the conditions for the circular orbit are
\begin{eqnarray}
\dot{r}&=&0,\\
V'(r)&=&-\frac{1}{2}\left[\left(\frac{1}{fg}\right)'F+\frac{1}{fg}F'\right]=0.
\end{eqnarray}
The former gives $V(r)=0$ from %Eq.~(\ref{eq:kplusv})
Eq.~(\ref{eq:hamiltonian}) and so $F(r)=0$ from $1/(fg) \ne 0$.
Then the latter implies $F'(r)=0$ and the radius of the photon spheres is specified by the condition
\begin{equation}
(fr^{-2})'=0.
\end{equation}
Circular orbits are classified into two kinds, stable and unstable orbits.
They correspond to the conditions $V''(r)>0$ and $<0$, respectively, at the radius.
Using the fact $F=F'=0$ at the radius, we have
\begin{equation}
V''(r)=-\frac{1}{2}gf^{-1}F''(r).
\end{equation}
Thus the (in)stability is established by $(fr^{-2})''>0\ (<0)$
uniquely and we get Eq.~(\ref{eq:spherestability}).
\qedhere
\end{proof}
The conditions do not depend on the component $g(r)$ of the metric. 

%%%chap3
\section{The accretion problem in $D$ dimensional spacetime and its critical point and sonic point}
\llaabel{sec:accretionproblem}
Here, assuming three conservation laws and the metric~(\ref{eq:metric}), the formulation of the accretion problem is given.
The definitions of the critical point and the sonic point are also given in the subsequent subsections.
\par
We assume three conservation equations, i.e., the first law, continuity equation and energy-momentum conservation with perfect fluid:
\begin{subnumcases}
{}
\llaabel{eq22a}
dh=Tds+n^{-1}dp\\
\llaabel{eq22b}
\nabla_\mu J^\mu=0\\
\llaabel{eq22c}
\nabla_\mu T^\mu_\nu=0,
\end{subnumcases}
where $J^\mu:=nu^\mu$ is the number current and $T^\mu_{\nu}=nhu^\mu
u_\nu+p\delta^\mu_\nu$ is the energy-momentum tensor of the perfect fluid.
The quantities $h,T,s,n,p$ and $u^\mu$ represent the enthalpy per particle, the temperature,
the entropy per particle, the number density, the pressure and the
4-velocity of the fluid, respectively.
The system of the equations means the adiabatic condition of the fluid 
through Eqs.~(\ref{eq22a}), (\ref{eq22b}) and (\ref{eq22c})
multiplied by $u^\nu$.
Furthermore, the stationary and spherically symmetric state of the flow
implies that the entropy is constant over the whole spacetime,
allowing us to write $h=h(p)$ or
\begin{equation}
\llaabel{eq24}
h=h(n).
\end{equation}
Integrating Eq.~(\ref{eq22b}), we have
\begin{equation}
\llaabel{eq25}
j_n:=4\pi(fg)^{1/2}r^{D-2}nu^r=const,
\end{equation}
from the symmetry of the fluid and the spacetime metric 
(\ref{eq:metric}).
The quantity $j_n$ represents the particle flux of the fluid.
The component, which is independent from $u^\mu$, of Eq. (\ref{eq22c}) is obtained
by 
multiplying it by the static Killing vector
$\xi_{(t)}^\nu:=\delta^\nu_t$. The integration of this component gives
\begin{equation}
\llaabel{eq26}
j_\epsilon:=4\pi(fg)^{1/2}r^{D-2}nhu^r\sqrt{f+fg(u^r)^2}=const,
\end{equation}
for the energy flux.
Combining Eqs.~(\ref{eq24}), (\ref{eq25}) and (\ref{eq26}), we get
\begin{eqnarray}
\left(\frac{j_\epsilon}{j_n}\right)^2&=&h^2\left[f+fg(u^r)^2\right]\\
&=&h^2(n)\left[f(r)+\frac{(j_n/4\pi)^2}{r^{2(D-2)}n^2}\right]=const.
\end{eqnarray}
Then, defining the constant $\mu:=j_n/4\pi$,
the problem is formulated into the algebraic equation:
\begin{equation}
\llaabel{eq23}
F_{\mu}^{(D)}(r,n):=h^2(n)\left[f(r)+\frac{\mu^2}{r^{2(D-2)}n^2}\right]=const.
\end{equation}
The physical meaning of $\mu$ is an accretion rate.
Given $\mu$, the function is specified and the constant on the RHS of Eq.~(\ref{eq23}) determines an accretion flow.
Note that this equation does not depend on the 
$rr$-component $g_{rr}$ of the metric.

\subsection{The critical point}
From the system (\ref{eq23}), the stationary accretion solutions are described as curves on
the phase space $(r,n)$.
These curves can be obtained by integrating the ordinary differential equation,
\begin{equation}
\llaabel{eq30}
\frac{d}{d\lambda}\left(\begin{array}{c} r\\ n \end{array}\right)
=\left(\begin{array}{r} \partial_n \\ -\partial_r \end{array}\right)F_\mu^{(D)}(r,n),
\end{equation}
as orbits with a parameter $\lambda$.
Then a notion of a {\it critical point} (or stationary point as in dynamical systems) at which the RHS of Eq.~(\ref{eq30}) vanishes arises
and its conditions are
\begin{subnumcases}
{}
\partial_nF_\mu^{(D)}=0\\
\partial_rF_\mu^{(D)}=0.
\end{subnumcases}
These are equivalent to
\begin{subnumcases}
{}
\llaabel{eq32a}
v_s^2\left(f+\frac{\mu^2}{r^{2(D-2)}n^2}\right)-\frac{\mu^2}{r^{2(D-2)}n^2}=0\\
\llaabel{eq32b}
f'-\frac{2(D-2)}{r}\frac{\mu^2}{r^{2(D-2)}n^2}=0,
\end{subnumcases}
respectively, where the sound speed $v_s=v_s(n)$ is defined by
\begin{equation}
v_s^2:=\frac{\partial \ln h}{\partial \ln n}.
\end{equation}
In the following, $(r_c,n_c)$ denotes the critical point.

\subsubsection{Types of critical points}
\llaabel{sec:types}
The linearization of Eq.~(\ref{eq30}) around a critical point allows us to classify the critical point into two types.
The one is a saddle point and the other is an extremum point.
A saddle point is a point, in this case, through which two solution orbits pass.
On the other hand, orbits in vicinity of an extremum point are closed curves around the point.
\par
The linearization matrix $M_c$ is given by
\begin{equation}
M_c:=\left(
	\begin{array}{cc}
	\partial_r\partial_n & \partial_n^2 \\
	-\partial_r^2 & -\partial_r\partial_n 
	\end{array}
	\right)F_\mu^{(D)}(r_c,n_c).
\end{equation}
This matrix, being real, $2\times 2$ and traceless, has two eigenvalues with opposite signs.
The subscript $c$ denotes the values at $(r_c, n_c)$.
If the determinant of the matrix is negative (positive), the eigenvalues
are real (pure imaginary).
As in dynamical systems, real eigenvalues imply that the critical point is a saddle point.
For imaginary eigenvalues, the orbits around the critical point are periodic in linear order.
However, because they are the contours of the real function $F_\mu^{(D)}(r,n)$,
the orbits must be closed loops.
Therefore the imaginary eigenvalues imply an extremum point.
We can write the determinant explicitly,
\begin{equation}
\det M_c=-\frac{2}{D-2}r_c(f'_c)^2\frac{h_c^4}{n_c^2}{\mathcal{F}_\mu^{(D)}}'(r_c),
\end{equation}
where
\begin{eqnarray}
\mathcal{F}_\mu^{(D)}(r)&:=&v_s^2(n_D(r))\left[1+2(D-2)a(r)\right]-1,\nonumber\\
n_D(r)&:=&\sqrt{\frac{D-2}{2}}\frac{2|\mu|}{\sqrt{r^{D+1}f'(r)}},\nonumber\\
a(r)&:=&\frac{f(r)}{rf'(r)}.\nonumber
\end{eqnarray}
Then we have a simple relation:
\begin{equation}
\llaabel{eq:saddleextremum}
saddle\ (extremum)\ point \Leftrightarrow {\mathcal{F}_\mu^{(D)}}'(r_c)>0\ (<0)
\end{equation}

\subsection{The sonic point}
Although a critical point is a
mathematical notion
defined on the phase space $(r,n)$ 
of a dynamical system, 
this also is closely related to a physical entity, a sonic point.
We define a sonic point and see its relation with a critical point in the following.
\subsubsection{The transonic flow and the sonic point}
In an accretion problem, one may expect that the fluid element at infinity,
which falls with small 3-velocity, becomes faster and faster as approaching the source of the gravity.
If the acceleration is sufficient, the velocity, initially smaller than its local sound speed $v_s$ (subsonic) at infinity, would become greater than $v_s$ (supersonic) at the point near the source.
Such a fluid flow is said to be {\it transonic} and here we call any flow which has both sub- and supersonic regions transonic.
Since, in our accretion problem, a fluid accretion flow is a solution orbit of Eq.~(\ref{eq23}),
we define a {\it sonic point} of a transonic flow as follows.
\begin{definition}
For a stationary and spherically symmetric accretion flow on the
 spacetime metric (\ref{eq:metric}), let $n=n(r)$ be its corresponding solution orbit on the phase space $(r,n)$.
Let $v=v(r)$ be the radial component of the 3-velocity of the fluid measured by static observers.
A sonic point $(r_s,n_s)$ of the accretion flow is defined as a point on the phase space satisfying the condition,
\begin{eqnarray}
\left.\frac{v^2}{v_s^2}\right|_{(r_s,n(r_s))}=1,
\end{eqnarray}
where $n_s=n(r_s)$.
\end{definition}
\subsubsection{The sonic point and the critical point}
The critical point mentioned above is closely related to the sonic point
and we present a lemma.
\begin{lemma}\llaabel{lemma:sonic-critical}
Assume the EOS of the fluid satisfies the condition
\begin{eqnarray}
\llaabel{eq:subluminal}
&0<v_s^2(n)<1,\\
\llaabel{eq:monoinc}
&\partial_n v_s^2(n)\ge0.
\end{eqnarray}
That is, the sound speed of the fluid is subluminal and monotonically increasing with respect to $n$.
For a physical transonic accretion flow in our accretion problem, its sonic point coincides with a critical point on
the phase space, which is a saddle point.
\end{lemma}
\begin{proof}
For the flow, the radial component of its 3-velocity $v(r)$ observed by static observers is given by
\begin{equation}
u^\mu\partial_\mu=\frac{1}{\sqrt{1-v^2}}(e_0+ve_1), 
\end{equation}
where $e_0:=f^{-1/2}\partial_t$ is the observers' 4-velocity and $e_1:=g^{-1/2}\partial_r$ is
a unit radial vector orthogonal to it.
Then, we have
\begin{equation}
v^2(r)=\frac{\mu^2}{\mu^2+f(r)r^{2(D-2)}n^2(r)}
\end{equation}
along the orbit $n=n(r)$ using Eq.~(\ref{eq25}), $-1=u^\mu u_\mu$ and $\mu=j_n/4\pi$.
On the other hand, letting $n=\tilde{n}(r)$ be a curve satisfying the condition $\partial_nF_\mu^{(D)}=0$, or equivalently Eq.~(\ref{eq32a}), we have the relation
\begin{equation}
\llaabel{eq:vsntilde}
v_s^2(\tilde{n}(r))=\frac{\mu^2}{\mu^2+f(r)r^{2(D-2)}\tilde{n}^2(r)}
\end{equation}
for the sound speed $v_s$.
From the two equations above and the assumption $\partial_n v_s^2\ge0$,
if $n(r_0)>\tilde{n}(r_0)$ for radius $r=r_0$, 
$v^2(r_0)<v_s^2(\tilde{n}(r_0))\le v_s^2(n(r_0))$, i.e., subsonic.
In the same way, the flow is supersonic at the radius if
 $n(r_0)<\tilde{n}(r_0)$. This means
that the curve $n=\tilde{n}(r)$ divides the phase space into subsonic and supersonic region and
the sonic point must be the point at which the orbit $n=n(r)$ and the curve $n=\tilde{n}(r)$ cross each other.
(Conversely, such a crossing point must be the sonic point of the flow.)
However, if $\partial_rF_\mu^{(D)}\neq0$ at the crossing point, such an
 orbit typically gets 2-valued (so unphysical)
at least locally because
$dn/dr=\partial_rF_\mu^{(D)}/\partial_nF_\mu^{(D)}=\pm\infty$ there from
 Eq.~$(\ref{eq30})$. 
%Such behavior is not permitted for a physical flow.
In the current paper, we require $|dn/dr|<\infty$ as 
one of the conditions of a physical flow.
Then, it is said that physically acceptable transonic orbits cross the curve of $\partial_nF_\mu^{(D)}=0$ only at a critical point and so the sonic point coincides with the critical point.
Furthermore, according to the discussion in Sec.~ \ref{sec:types}, the critical point is a saddle point because orbits can pass the point.
Finally, we must show that the function $\tilde{n}(r)$ is indeed single-valued.
We can separate Eq.~(\ref{eq:vsntilde}) into a function of $\tilde{n}$ and the rest,
\begin{eqnarray}
&\mathcal{N}\left(\tilde{n}(r)\right)=\mu^{-2}f(r)r^{2(D-2)},\nonumber\\
&\mbox{where}~~
\mathcal{N}(n):=\left(v_s^{-2}(n)-1\right)n^{-2}.
\end{eqnarray}
The conditions (\ref{eq:subluminal}) and (\ref{eq:monoinc}) imply
$\partial_n\mathcal{N}(n)<0$, $\mathcal{N}(n)\to0\ (n\to\infty)$ and $\mathcal{N}(n)\to\infty\ (n\to0)$.
Therefore, the inverse function $\mathcal{N}^{-1}:(0,\infty)\to(0,\infty)$ exists and
$\tilde{n}(r)$ can be expressed as a single-valued function,
\begin{equation}
\tilde{n}(r)=\mathcal{N}^{-1}\left(\mu^{-2}f(r)r^{2(D-2)}\right).
\end{equation}
\end{proof}

%%%chap4
\section{The photon gas accretion and its critical point}
\llaabel{sec:photongas}
In this section, we will construct the accretion problem of ideal photon
gas in $D$ dimensions
and find the condition of its critical point based on discussions in the previous section.

\subsection{The EOS of ideal photon gas in $d$ dimensional space}
To formulate the accretion of ideal photon gas in $D$ dimensions, we must know its equation of state
at first.
Here we construct the EOS.
\par
From the discussion of black body radiation in a $d$ dimensional space, we have a relation
\begin{equation}
pV=\frac{1}{d}U,
\end{equation}
where the thermodynamical variables $p, V$ and $U$ are the pressure, the volume and the energy of a system, respectively.
This relation gives
\begin{equation}
\left(\frac{\partial U}{\partial V}\right)_S\equiv-p=-\frac{1}{d}\frac{U}{V},
\end{equation}
where $S$ denotes the entropy.
Integrating the both sides concerning $U$ and $V$
\begin{equation}
UV^{1/d}=C(S),
\end{equation}
with the function $C(S)$ being an arbitrary function.
Then, the enthalpy $H$ of the black body radiation is
\begin{equation}
H=U+pV=\frac{d+1}{d}U\propto V^{-1/d}.
\end{equation}
Note that the proportionality coefficient of the last equality can depend on the entropy $S$.
Comparing this result with the usual convention of a polytrope index in the expression per particle, we conclude that the EOS of ideal photon gas is
\begin{equation}
\llaabel{eq37}
h=\frac{k\gamma}{\gamma-1}n^{\gamma-1}
\end{equation}
with
\begin{equation}
\llaabel{eq:gamma}
\gamma=\frac{d+1}{d}
\end{equation}
and $k$ is an arbitrary function of the entropy.
It can be revealed that the quantity $k$ is a constant constructed by
the Planck constant and a numerical factor by the argument about photon gas from statistical mechanics.
However, the explicit form of $k$ is not relevant to the proof of the theorem. Since the entropy of the fluid is constant over the spacetime, $k$ is also constant.
This is relevant to the proof.

\subsection{The critical point of photon gas accretion}
\begin{lemma}\llaabel{lemma:photongascritical}
For the accretion of ideal photon gas in our accretion problem,
 the
radius $r_c$ of a critical point
is specified by
\begin{equation}
\llaabel{eq45}
(fr^{-2})'=0
\end{equation}
and the corresponding critical density $n_c$ is
\begin{equation}
n_c=\sqrt{\frac{D-2}{f(r_c)}}\frac{\left|\mu\right|}{r_c^{D-2}}.
\end{equation}
The type of the critical point is classified by the equation
\begin{equation}
saddle\ point\ (extremum\ point)\ \Leftrightarrow (fr^{-2})''<0\ (>0)
\end{equation}
at the radius.
\end{lemma}

\begin{proof}
The condition for a critical point $(\ref{eq32a})$, $(\ref{eq32b})$ can be transformed to
\begin{subnumcases}
{}
\llaabel{eq46a}
v_s^2\left[2(D-2)f+rf'\right]-rf'=0\\
\llaabel{eq46b}
f'-\frac{2(D-2)}{r}\frac{\mu^2}{r^{2(D-2)}n^2}=0.
\end{subnumcases}
From Eq.~$(\ref{eq37})$, the sound speed of ideal photon gas is constant,
\begin{equation}
v_s^2=\frac{\partial \ln h}{\partial \ln n}=\gamma-1.
\end{equation}
Substituting this, Eq.~$(\ref{eq46a})$ determines the position $r$ of the
 critical point,
\begin{equation}
0=(\gamma-2)r^3(fr^{-2})',
\end{equation}
where the formula $(\ref{eq:gamma})$ and the fact $d=D-1$ are used in the last equality.
The corresponding number density at the critical point is uniquely obtained
from Eq.~$(\ref{eq46b})$ using the relation $f'(r_c)=2f(r_c)/r_c$.
The condition for a saddle or extremum point was given in \ref{sec:types}.
In this case, the value of the function ${\mathcal{F}_\mu^{(D)}}'$ at a
 critical point is written in the form,
\begin{equation}
{\mathcal{F}_\mu^{(D)}}'=-2(\gamma-1)(D-2)rf(f')^{-2}(fr^{-2})'',
\end{equation}
where the fact $(fr^{-2})'=0$ at a critical point is used.
Clearly, the sign of the value ${\mathcal{F}_\mu^{(D)}}'$ is opposite to $(fr^{-2})''$ and the proof has been done.
\end{proof}

Note that the sound speed $v_s$ satisfies the subluminal condition $0<v_s^2<1$ and the monotonically increasing condition $\partial_nv^2_s\ge0$ from
$v_s^2=\gamma-1=1/d$ and $d=D-1\ge2$.

%%%chap5
\section{The proof of Theorem: The correspondence among the points}
\llaabel{sec:proof}
In this section, we see the correspondence among the three objects; the photon sphere, the critical point and the sonic point of our ideal photon gas accretion problem and complete the proof of the main theorem.
\par
From Lemma \ref{lemma:photonsphere} and \ref{lemma:photongascritical}, we can establish
the following corollary about the correspondence between the 
photon spheres and the critical points of ideal photon gas accretion.
\begin{corollary}
\llaabel{corollary:spherecritical}
If the spacetime has photon spheres, there exists a critical point of the same radius for each of the spheres.
Furthermore, for an unstable photon sphere, the critical point on the same radius is always a saddle point while for a stable one, the corresponding critical point is an extremum point.
\end{corollary}
The critical point radius $r_c$ depends on $\mu$ in general.
However, the photon gas accretion is interesting in the sense that its critical radius does not depend on $\mu$ and this fact is responsible for the correspondence.
\par
Ideal photon gas satisfies the condition of EOS in Lemma \ref{lemma:sonic-critical}
since $v_s^2(n)=\gamma-1$ and the lemma can be applied.
Then we have the corollary about the relation between critical points and sonic points.
\begin{corollary}
\llaabel{corollary:criticalsonic}
For any physical transonic accretion flow of the ideal photon gas accretion, its sonic point is a critical and saddle point.
\end{corollary}
Then, the above two corollaries complete the proof of Theorem.

%%%sec
\section{Conclusion}
\llaabel{sec:conclusion}
In this work, first we derived the conditions for photon spheres, the radius and the stability of the corresponding circular orbit of null geodesics.
Next, we generalize the accretion analysis given by E. Chaverra and O. Sarbach~\cite{chaverra}
to arbitrary dimensions and discussed the relation between sonic points and critical points in general.
Then, for ideal photon gas, it was shown that radius of a sonic point always coincides with (one of) photon spheres for physical solutions of the accretion problem.
\par
We can say that a photon sphere is indeed special even for the radial accretion
because the flow can be interpreted as a set of geodesic motions of photons
and some of the photons must go round on the sphere.
However, the sound speed and the fluid velocity are macroscopic quantities.
The reason for the correspondence is not yet so
clear. 
Since the correspondence seems to originate from the
microscopic construction of radiation fluid, we conjecture that the
correspondence will be seen in more general situations, such as 
axially symmetric steady-state accretion flows onto stationarily 
rotating black holes.
\par
The correspondence can be broken if the effects of the back reaction is included.
However, since photon spheres are usually located near the source of
gravity, it would be justified to neglect the self-gravity of the fluid
and the correspondence still holds in that case. 
\par
 It should be noted that the present discussion applies not only to accretion but also to outflow or stellar wind as long as it is steady-state and spherically symmetric.

\begin{acknowledgments}
We thank T. Igata, M. Patil, T. Kokubu, K. Ogasawara,
 S. Kinoshita, H. Maeda and G.W.~Gibbons for their very helpful
 discussions and comments.
The authors are grateful to the anonymous referee
for his/her valuable comments that improved the
manuscript.
TH was supported by JSPS KAKENHI Grant Number 26400282.
\end{acknowledgments}

%\bibliography{sample}
%
\bibliography{manuscriptsonicphoton}

\end{document}